\documentclass[a4paper,UKenglish,cleveref, autoref, thm-restate]{lipics-v2021}

\title{Table-based Quantifier Elimination}
\titlerunning{Table-based Quantifier Elimination}

\author{Pierre Carbonnelle}{Independent Researcher, Brussels, Belgium}{pierre.carbonnelle@gmail.com}{https://orcid.org/0000-0002-7074-1870}{}

\authorrunning{Pierre Carbonnelle}

\date{January 2026}
\Copyright{Pierre Carbonnelle}

\category{Long paper} 

\relatedversion{} 

\acknowledgements{I thank Joost Vennekens for his suggestion to use relational database management systems for grounding,
and Maurice Bruynooghe for his comments on early drafts.}

\nolinenumbers

\EventEditors{John Q. Open and Joan R. Access}
\EventNoEds{2}
\EventLongTitle{42nd Conference on Very Important Topics (CVIT 2016)}
\EventShortTitle{CVIT 2016}
\EventAcronym{CVIT}
\EventYear{2016}
\EventDate{December 24--27, 2016}
\EventLocation{Little Whinging, United Kingdom}
\EventLogo{}
\SeriesVolume{42}
\ArticleNo{23}

\usepackage{amssymb}

\usepackage{unicode-math}
[    Extension = .otf,
      BoldFont = XITSMath-Bold,
]
\newcommand\fullouterjoin{\fullouterjoin}

\usepackage{graphicx} 

\usepackage{amsmath}
\usepackage{booktabs}
\usepackage{cancel}
\usepackage{cellspace} 
\usepackage{dsfont}
\usepackage{fontspec}
\usepackage{graphicx}
\usepackage{makecell}
\usepackage{pdflscape}
\usepackage{stmaryrd}
\usepackage{threeparttable}
\usepackage{tikz}
\usepackage{todonotes}
\usetikzlibrary{patterns} 
\usepackage{xcolor}
\usepackage{xparse}
\usepackage{xspace}

\newcommand\mm[1]{\ensuremath{\mathit{#1}}\xspace}
\newcommand{\proj}[1]{\Pi_{\textstyle #1}\,}

\newcommand{\define}[1]{\colorbox{lightgray}{\textit{#1}}}

\newcommand{\snot}[0]{\mm{NOT}}
\newcommand{\sor}[0]{\mm{OR}}
\newcommand{\sand}[0]{\mm{AND}}
\newcommand{\sequals}[0]{\underline{=}}
\newcommand{\slt}[0]{\underline{<}}
\newcommand{\sleq}[0]{\underline{\leq}}
\newcommand{\sgt}[0]{\underline{>}}
\newcommand{\sgeq}[0]{\underline{\geq}}

\begin{document}

\maketitle
\begin{abstract}
    Quantifiers make first-order logic more expressive than propositional logic, but they also make solving satisfiability problems more difficult.  To solve these problems efficiently, many techniques have been developed to eliminate quantifiers from logic formulas: model-based quantifier instantiation, E-matching, Skolemization, destructive equality resolution, ...

    I propose a table-based instantiation method for quantifier elimination.
    It can be used to solve satisfiability problems for ``guarded'' formulas, i.e., formulas where the use of defined functions in the matrix of quantifications effectively restricts the relevant range of their variables to a finite set.

    I have implemented this method in a pre-processor for SMT solvers.  This ``grounder'' leverages an embedded relational database (SQLite) to execute relational operations efficiently.  While this grounder does not accelerate 2025 benchmarks of the SMT-LIB library, it does improve performance of SMT solvers on public benchmarks for SAT and ASP solvers.  The grounder has applications in, e.g., solving graph problems.
\end{abstract}

\section{Introduction}


Quantifiers make first-order logic more expressive than propositional logic, but they also make solving satisfiability problems more difficult.  To solve these problems efficiently, many techniques have been developed to eliminate quantifiers from logic formulas: model-based quantifier instantiation, E-matching, Skolemization, destructive equality resolution, ...

This work investigates quantifier elimination for satisfiability problems when the use of defined functions in the matrix (i.e., the sub-formula following the quantifier) of quantifications effectively restricts the relevant range of the quantified variables to a finite set.
Graph problems are good examples: statements about edges are encoded as formula quantified over pairs of nodes, but restricted to the pairs that are edges.

\begin{example}
    As a simple illustration, consider the graph-coloring problem for a directed graph with two edges,
    $1 \rightarrow 2$ and $2 \rightarrow 3$, where $1, 2, 3$ denote nodes.
    A natural specification introduces an  \texttt{edge} predicate
    and a \texttt{colorOf} function.
    The interpretation of \texttt{edge} is known: it can be encoded in the SMT-LIB language~\cite{BarFT-RR-25} using a definition:
    \begin{lstlisting}
        (define-fun edge ((x Int) (y Int))
            (or (and (= x 1) (= y 2))
                (and (= x 2) (= y 3)) ))
    \end{lstlisting}

    The assertion that nodes connected by an edge cannot have the same color is:
    \begin{lstlisting}
        (assert (forall ((x Int) (y Int))
                    (=> (edge x y)
                        (not (= (colorOf x) (colorOf y))))))
    \end{lstlisting}
    The definition of \texttt{edge} effectively restricts the relevant range of the quantification to $(1,2)$ and $(2,3)$:
    for any other substitution, the matrix of the quantification is true because the premise \texttt{(edge x y)} is false.\lipicsEnd
\end{example}

In our experiments, the cvc5 solver cannot solve this problem when the graph is larger than 50, in any encoding that uses a quantifier.
On the other hand, the Z3 solver can solve it, and Model-Based Quantifier Instantiation is fastest.  However, it does not scale well, per our experiments.



In this paper, I present a method that exploits the definitions of user-declared function symbols to derive, when possible, a finite set of substitutions for the expansion of quantified formulas.
Further, when several such defined symbols occur in a quantified formula, it uses relational algebra to further reduce the set of substitutions.
I call this method ``table-based quantifier elimination''.
The hypothesis is that using it can accelerate solving when the use of defined functions in the matrix of quantifications effectively restricts the relevant range of the quantifications to a finite set.

To test this hypothesis, I have implemented this method in a ``grounder'' called \texttt{xmt-lib}.
It is a preprocessor that transforms an SMT-LIB script into an equisatisfiable script in a Quantifier Free (QF) logic, when possible.
In our example, it replaces the quantified formula by the finite number of assertions shown in \cref{lst:assertions}:

\begin{lstlisting}[caption={Assertions generated by xmt-lib for the graph coloring example.}, label={lst:assertions}]
    (assert (not (= (colorOf 1) (colorOf 2))))
    (assert (not (= (colorOf 2) (colorOf 3))))
\end{lstlisting}

The implementation leverages a relational database engine (SQLite)
to perform the required relational operations efficiently,
and can directly exploit input data already stored in the database.
A demo version is available online.\footnote{https://pcarbonn.github.io/XMT-IDE}

I evaluate \texttt{xmt-lib} on the graph coloring problem above, and on benchmarks specifically developed to evaluate grounders~\cite{DBLP:conf/jelia/LaerVV25}.
The results show that \texttt{xmt-lib} can significantly improve
the performance of SMT solvers.
For example, it solves a randomly-generated graph coloring problem with 2500 nodes and an edge density of 1 per cent in less than 1 second, while SMT solvers take more than 10 minutes without it.
Also, it makes SMT solving competitive with state-of-the-art ASP solvers
for the problems in this benchmark.

However, the grounder does not improve the performance of 2025 benchmarks in the SMT-LIB suite.  I have thus submitted new benchmarks for the class of problems where the grounder brings benefits.


In this paper, after describing general concepts and notations of first-order logic and relational algebra, I explain table-based quantifier elimination using relational algebra.  After describing the implementation, I report on its evaluation, and discuss related work.



\section{Preliminaries}

\subsection{SMT-LIB}

I present the logic using the terminology and ontology of the SMT-LIB standard~\cite{BarFT-RR-25}.
The logic is many-sorted.
A \define{signature} $\Sigma$ consists of:
(1) a set of \define{sort symbols} (including $\mathbf{Bool}$), denoted by $\sigma$;
(2) a set of ranked \define{function symbols} $f{:}\sigma_1 \dots \sigma_n \to \sigma$, or $f$ for short;
(3) a set of \define{sorted variables} $x{:}\sigma$, or $x$ for short.
A function symbol with result sort $\mathbf{Bool}$ is a \define{predicate}.

Function symbols are either \define{interpreted}, \define{defined}, or \define{uninterpreted}.
Interpreted symbols have a fixed behavior and belong to \define{background theories} (e.g., arithmetic).
A \define{definition} (using \texttt{define-fun}) associates a function symbol with a body term; such a symbol is defined.
Symbols that are neither interpreted nor defined are uninterpreted.

A \define{well-sorted term} is constructed from variables and function symbols applied to terms of the expected sorts.
A term is \define{ground} if it contains no variables.
A term $f(\bar t)$ is \define{nested} if it appears as an argument of another function or predicate.
A term $f(\bar t)$ is \define{flat} if its arguments are either variables or ids: it has no nested terms.
A term of sort $\mathbf{Bool}$ is a \define{formula}.
A closed formula is a \define{sentence}.
A \define{quantification} $\forall \bar{x} : \phi$ or $\exists \bar{x} : \phi$ consists of a quantifier prefix binding a set of variables $\bar{x}$, and a sub-formula $\phi$, called the \define{matrix} of the quantification.
A \define{literal} is either a predicate application $p(\bar t)$ or its negation $\neg p(\bar t)$.
A \define{clause} is a disjunction of literals.

Some interpreted symbols are \define{constructors}.
A \define{ground constructor term}, or \define{id} for short, is either a parameterless constructor (such as $\top$, $\bot$, or integers) or a constructor applied to ids.
Each constructor $c$ with arity $n$ is associated with $n$ \define{selectors} $s_1, \dots, s_n$ such that for any id $e = c(e_1, \dots, e_n)$, the value of $s_i(e)$ is $e_i$.

In the semantics of the logic, each sort $\sigma$ is associated with a \define{universe} $U_\sigma$.
A sort $\sigma$ is \define{representable} if its universe $U_\sigma$ is exactly its set of ids; for such a sort, the value of any ground term is an id.
For instance, $\mathbf{Bool}$ and $\mathbf{Int}$ are representable sorts, whereas $\mathbf{Real}$ is not.
The $\mathbf{Bool}$ ids are represented in SMT-LIB by the constants \texttt{true} and \texttt{false} (denoted $\top$ and $\bot$).

A \define{substitution} $\bar v$ maps variables to ids.
Applying a substitution $\bar v$ to a term $t$ yields an \define{instantiation} $t[\bar v]$; it is ground if $\bar v$ is for all the free variables of $t$.

A \define{structure} $\mathbf{A}$ provides an interpretation for every symbol in the signature, where the interpretation of defined and interpreted symbols is fixed by their definitions and theory specifications, respectively.
Every ground term $t$ has a \define{value} in $\mathbf{A}$, denoted by $t^\mathbf{A} \in U_\sigma$.
A ground formula $\phi$ is \define{true in $\mathbf{A}$} if $\phi^\mathbf{A} = \top$.
In practice, the value of some ground terms can be determined by \define{simplification}: replacing defined symbols with the body of their definitions and applying theory rules.

\subsection{Relational algebra}

A \define{relation} is a pair consisting of (1) a set of tuples of elements in a domain, and (2) a tuple of \define{attribute names}.
An \define{attribute} in a relation is a column; a tuple of elements is a row.
The number of rows may be infinite.
Here, the domain of the relations is the set of terms introduced in Section~2.1.

Further, I distinguish \define{comparison expressions} and \define{construction expressions} over a relation.
The base symbols of both types of expressions are the attribute names of the relation, which act as variables.
Comparison expressions evaluate to a boolean value ($\top$ or $\bot$) for each row of the relation. Their operators, $\sequals, \slt, \sleq, \sgeq, \sgt, \snot, \sand, \sor$, map to the native operators of the underlying database.
Construction expressions build new terms from attribute values, one for each row of the relation; these terms are simplified when possible using the laws of the background SMT theories. Their operators are the function symbols of the signature (e.g., $+, =, <, \dots, \lnot, \land, \lor$).

I use the following notations for relations:
\begin{itemize}
    \item $0_{\bar n}$ is the \define{empty relation} with attributes $\bar n$ and no tuples; $1$ is the \define{unit relation} (the 0-relation with one empty tuple);

    \item $\sigma_{\theta} R$ is the \define{selection} of tuples of $R$ for which condition $\theta$ evaluates to $\top$ ($\theta$ is a comparison expression);

    \item $R \times S$ is the \define{cross-product} of relations $R$ and $S$; its attributes are the concatenation of those of $R$ and $S$ (denoted $\bar n, \bar m$), and its rows are the concatenations of every row of $R$ with every row of $S$;

    \item $R \bowtie S$ is the \define{natural join} of $R$ and $S$; a tuple over the union of attribute names of $R$ and $S$ is in $R \bowtie S$ if its restriction to the attributes of $R$ is in $R$ and its restriction to the attributes of $S$ is in $S$;

    \item $R \leftouterjoin_{\theta} S$ is the \define{theta left outer join} of $R$ and $S$; it contains the theta join $\sigma_{\theta} (R \bowtie S)$ plus the tuples in $R$ that have no matching tuple in $S$ according to $\theta$, extended with a special $\mm{null}$ value for each attribute of $S$;

    \item $R \cup S$ is the \define{union} of two relations with the same attribute names;

    \item $\proj{\bar n, m{:}e} R$ is the \define{extended projection} of $R$, where $\bar n$ is a list of attribute names of $R$, $m$ is a new attribute name, and $e$ is a construction expression; the result has attributes $\bar n, m$ and its rows are obtained by concatenating the values of $\bar n$ in each row of $R$ with the term obtained by evaluating $e$ on that row (which is simplified when possible);

    \item $\proj{\bar n, m{:} \mm{agg}(\mm{op}, e)} R$ is the \define{aggregation} of $R$ over attributes $\bar n$, where $m$ is a new attribute name and $\mm{op}$ is $\land$ or $\lor$; for each group of tuples in $R$ sharing the same values for $\bar n$, it produces a tuple where $m$ is the result of applying $\mm{op}$ to all terms obtained by evaluating $e$ on each tuple in the group (and simplifying the result).
\end{itemize}

Two sets $X, Y$ of attribute names in a relation have a \define{functional dependency} when the values of $X$ attributes in a tuple determine the values of $Y$ attributes.

\section{Table-based quantifier elimination} 

In this section, I present the table-based quantifier elimination method. It relies on \textit{look-up tables} to represent the interpretation of defined function symbols (Section 3.1). The method proceeds in two phases:
first, a set of \textit{substitution relations} is discovered \textit{upwards} from the leaves of the formula, starting with an account for clauses (Section 3.2) and generalizing to arbitrary boolean terms (Section 3.3);
in the second phase, these relations are pushed \textit{downwards} from the root to perform the final quantifier elimination (Section 3.4).

\subsection{Look-up table with default}

\begin{definition}
    A \define{look-up table with default} for a defined function symbol $f{:}\sigma_1 \dots \sigma_n \to \sigma$
    is a pair $(R, d)$ where $R$ is a finite relation with attribute names $a_1 \dots a_n, a$ whose values are ids of sorts $\sigma_1 \dots \sigma_n, \sigma$ respectively, and $d$ is an id of sort $\sigma$, such that:
    \begin{itemize}
        \item $R$ has a functional dependency from the first $n$ attributes to the last;
        \item for every tuple $(e_1, \dots, e_n, e) \in R$, $e \neq d$;
        \item the function $f$ is interpreted as $f(\bar e) = e$ if the tuple $(\bar e, e)$ is present in $R$, and $f(\bar e) = d$ otherwise.
    \end{itemize}
\end{definition}

I call the pair $(R, d)$ a \define{table} for short.
I denote a table for a function $f$ by $f^D = (R_f, d_f)$.
A defined function may have several tables, each for different defaults, when its domain is finite.
For a predicate, the default $d$ can be $\top$ or $\bot$.

This representation is compact because it only lists the exceptions to the default value. It is particularly useful for predicates where one of $\top$ or $\bot$ is much more frequent than the other.

In simple cases, a table can be derived from an SMT-LIB script by syntactic analysis, as in our introductory example.
\begin{example}[Graph Coloring Interpretation]
    \label{ex:interpretation}
    In our graph coloring example from the introduction, syntactic analysis of the definition of \texttt{edge} allows us to determine a table $(R, \bot)$. $R$ has
    attributes $(a_1, a_2, a)$
    and consists of the rows \texttt{(1, 2, $\top$)}, \texttt{(2, 3, $\top$)}.\lipicsEnd
\end{example}

In more complex cases, a table can be computed from a definition whose body only uses ids and function symbols having a table.

\begin{example}[Derived Tables]
    Consider another graph problem where \texttt{(edge x y)} is defined as true only when the \texttt{weight} of \texttt{x} is larger than the \texttt{weight} of \texttt{y}.  If the \texttt{weight} property has a table, a table for \texttt{edge} can be determined; otherwise, it cannot.\lipicsEnd
\end{example}

Note that a function may have a known mathematical interpretation without being representable as a table.
For example, the square root function in the theory of Reals ($sqrt: Real \to Real$) cannot be represented this way because it would require an infinite relation $R$, and some of its values are not representable as ids.

\subsection{Substitution relations for clauses}

To show the gist of the method, I begin by describing how an exhaustive set of substitutions can be determined for a universally quantified clause.
The next section will generalize the method to more complex formulas.

A substitution relation contains a set of substitutions:

\begin{definition}
    \label{def:substitution_relation}
    A \define{substitution relation} is a relation whose attribute names are sorted variable names and whose tuples are ids of the appropriate sorts.
    The \define{universal set of substitu-}
    \define{tions for a set of variables} $\bar x$, denoted by $X_{\bar x}$, is the Cartesian product of the universes of the sorts of these variables.
    The substitution relation \define{for a literal} $L$, denoted $S_L$, contains the set of substitutions of its free variables that, when applied to $L$, yield a ground literal that does not simplify to $\top$ using the look-up tables and background theory laws. 
\end{definition}
A substitution relation can have an infinite number of tuples.
The following definition and proposition justify the use of substitution relations for quantifier elimination:

\begin{definition}
    The \define{substitution relation for a clause} $C$, denoted $S_C$, is the relation over its free variables obtained by the natural join of the substitution relations of its literals, $S_C = ~\bowtie_i S_{L_i}$.
\end{definition}

\begin{proposition}
    \label{prop:clause_correctness}
    A universally quantified clause is equivalent to the conjunction of the clauses obtained by applying each substitution in its substitution relation to the clause.
\end{proposition}
\begin{proof}
    The implication $(\forall \bar x : C(\bar x)) \implies \bigwedge_{\bar v \in S_C} C(\bar v)$ is trivial.
    To prove the converse, suppose $\bigwedge_{\bar v \in S_C} C(\bar v)$ is true in a structure $\mathbf{A}$. I show that for any substitution $\bar u$ for the variables $\bar x$, $C(\bar u)$ is true in $\mathbf{A}$.
    If $\bar u \in S_C$, then $C(\bar u)$ is true by hypothesis.
    If $\bar u \notin S_C$, then by Definition~\ref{def:substitution_relation} and the definition of natural join, there must be at least one literal $L$ of $C$ such that the restriction of $\bar u$ to the free variables of $L$ (denoted $\bar u|_L$) is not in $S_L$.
    By Definition~\ref{def:substitution_relation}, this implies that the ground literal $L(\bar u|_L)$ simplifies to $\top$ using the look-up tables.
    Since $C$ is a disjunction of literals and one literal simplifies to $\top$, $C(\bar u)$ is true.
\end{proof}

A literal $L$ is a \define{guard} when its substitution relation $S_L$ is finite: a clause containing it is \define{guarded}. In this case, the substitution relation $S_C$ of the clause is restricted to a subset of $S_L \times X_{\bar y}$, where $X_{\bar y}$ is the universal set of substitutions for the other variables $\bar y$ of the clause. Indeed, by the definition of the natural join, for any tuple in $S_C$, its restriction to the variables of $L$ must belong to $S_L$.

\begin{definition}
    A variable of a quantified clause is \define{restricted} if it has a finite range in the substitution relation of the clause.
\end{definition}

A variable is restricted if its sort has a finite interpretation, or if it occurs in a guard literal.

Ultimately, the substitution relation $S_C$ is finite only when every variable in the clause is restricted. When it is infinite, the quantification of the restricted variables, if any, can still be expanded, while keeping the other variables universally quantified. However, the full expansion of the clause to a quantifier-free formula cannot be done, and other quantifier elimination techniques must be used for the non-restricted variables.

\subsubsection{Interpreted symbols.}
\label{sec:interpreted-symbols}

I begin by discussing the substitution relation of predicates of background SMT theories. Symbols from background SMT theories, such as equality, arithmetic operators, and comparisons, have fixed interpretations that the grounder exploits for simplification and restriction propagation.

For a literal $L$, $X_L$ denotes the universal set of substitutions for its free variables.

A literal is a \textbf{guard} if it restricts variables without requiring prior restrictions (\textbf{base}), or propagates variable restrictions \textbf{forward} (from restricted arguments to the output) or \textbf{backward} (from the restricted output to the arguments). Specifically, it is a guard in the following cases:
\begin{itemize}
    \item (\textbf{base}) if its universal set of substitutions $X_L$ is finite: its substitution relation is inherently finite;
    \item (\textbf{base}) if $L$ is a flat inequality of the form $x \neq e$ where $e$ is an id: its substitution relation is $\{e\}$;
    \item (\textbf{base}) if $L$ is of the form $f(\bar t) \neq a$, where $f$ is a constructor, $f(\bar t)$ is a flat term, and $a$ is an id: its substitution relation is the extended projection $\proj{\bar x{:} s_i(a)} (\sigma_\theta \{a\})$, where $\theta$ filters $a$ for the constructor $f$ and any ids or repeated variables in $\bar t$, while the attributes $\bar x$ are extracted via selectors for the distinct variables in $\bar t$;
    \item (\textbf{forward}) if $L$ is a literal $f(\bar t) \neq y$, $f(\bar t)$ is a flat term and all its variables are restricted: its substitution relation is the extended projection $\proj{\bar x, y{:} f(\bar t)} X_{\bar t}$.
    \item (\textbf{backward}) if $L$ is a literal $f(\bar t) \neq y$ where $f$ is a constructor, $f(\bar t)$ is a flat term, and $y$ is restricted by a substitution relation $S_y$: its substitution relation is the extended projection $\proj{\bar x{:} s_i(y), y} (\sigma_\theta S_y)$, where $\theta$ filters $y$ for the constructor $f$ and any ids or repeated variables in $\bar t$, while the attributes $\bar x$ are extracted via selectors for the distinct variables in $\bar t$;
\end{itemize}

\begin{example}[Restriction Propagation]
    Consider the flat literal \texttt{(not (= (cons x y) z))}.
    If the variable $z$ is restricted by the context, the literal is a backward guard that propagates the restriction on $z$ to the arguments $x$ and $y$ by applying selectors to $z$.
    Conversely, if the variables $x$ and $y$ are restricted by the context, the literal is a forward guard that propagates the restrictions on $x$ and $y$ to $z$ by evaluating the constructor.\lipicsEnd
\end{example}

A literal is a \textbf{filter} (its substitution relation $S_L$ is generally infinite, but restricts the join relation) in the following case:
\begin{itemize}
    \item if $L$ is a flat literal: its substitution relation is the theta selection $\sigma_\theta X_L$ where $\theta$ is the condition that $L$ does not simplify to $\top$ using theory laws (e.g., if $L$ is $x < 5$, $S_L = \sigma_{x \sgeq 5} X_L$); this rule covers ground literals (yielding relation $0$ or $1$) and general inequalities.
\end{itemize}
If the interpreted function $f$ does not fit these cases or any of its arguments is a nested term, the substitution relation is the universal set $X_L$.

In practice, the natural join of the substitution relations of all literals is simplified in two steps: (1) extracting and combining all selections into a single selection $\sigma_\Theta$ using the law $R \bowtie (\sigma_\theta S) = \sigma_\theta (R \bowtie S)$; and (2) removing the restricted variables $\bar x$ from the universal sets $X_{\bar x, \bar y}$ of non-guard literals using the law $S_{\bar x} \bowtie X_{\bar x, \bar y}= S_{\bar x} \times X_{\bar y}$, which leaves any unrestricted variables in a final cross product. The resulting formula has three parts:
\[ S_C = \sigma_\Theta \left(\bowtie_{i} S_i\right) \times X_u \]
where (i) $\bowtie_i S_i$ is a join of the finite relations provided by guard literals and extended projections for interpreted functions, (ii) $\sigma_\Theta$ is a theta selection for interpreted predicates and table filters, and (iii) $X_u$ is the universal set of substitutions for any unrestricted variables $u$. The first two parts are used to compute the substitution relation for the restricted variables of the clause.

\subsubsection{Defined and uninterpreted symbols.}
\label{sec:defined-symbols}

I now expand the previous discussion to include defined and uninterpreted symbols.

The table $(R, d)$ of the predicate of a literal $L$ \define{is appropriate} for $L$ if $d = \bot$ and $L$ is negative, or if $d = \top$ and $L$ is positive.

A literal is a \textbf{guard} in the following cases:
\begin{itemize}
    \item (\textbf{base}) if its universal set of substitutions $X_L$ is finite: its substitution relation is inherently finite;
    \item (\textbf{base}) if the predicate of $L$ has an appropriate table $(R, d)$ and $L$ is a flat literal: $S_L$ is derived from $R$ by (i) filtering the tuples to match any ids in the arguments, (ii) filtering the tuples to satisfy any variable repetitions, and (iii) projecting and renaming the attributes to match the variables of the literal;
    \item (\textbf{base}) if $L$ is of the form $f(\bar t) \neq a$ where $f$ has a table $(R, d)$, $f(\bar t)$ is a flat term, $a$ is an id, and $d \neq a$: the substitution relation $S_L$ is finite and contains the substitutions for which $f(\bar t) = a$ (which are exactly in $R$);
    \item (\textbf{forward}) if $L$ is a literal $f(\bar t) \neq y$ where $f$ has a table $(R, d)$, $f(\bar t)$ is a flat term, and all its variables are restricted: the substitution relation is derived by performing a left outer join between the finite set $X$ of substitutions of its variable arguments and the relation $R$. Similar to the table of a predicate, the result is filtered and renamed based on the arguments $\bar t$. Any null value for $y$ produced by the outer join is replaced by the default value $d$;
\end{itemize}

A literal is a \textbf{filter} (its substitution relation $S_L$ is generally infinite, but restricts the join relation) in the following cases:
\begin{itemize}
    \item if the predicate of $L$ has a table $(R, d)$ that is not appropriate for $L$ and $L$ is a flat literal: its substitution relation is the selection $\sigma_\theta X_L$, where $\theta$ is the condition that the arguments of the literal do not match any row in $R$;
    \item if $L$ is of the form $f(\bar t) \neq a$ where $f$ has a table $(R, d)$, $f(\bar t)$ is a flat term whose variables are all restricted, $a$ is an id, and $d = a$: the substitution relation is the selection $\sigma_\theta X_L$ where $\theta$ is the condition that the arguments do not match any row in $R$.
\end{itemize}
If the function $f$ does not fit these cases or any of its arguments is a nested term, the substitution relation is the universal set $X_L$.

\begin{example}[Variable Repetition and Constants]
    Consider the literal \texttt{(p 1 x x)} where \texttt{p} has table $(R, \bot)$ with $R = \{(1, 2, 2, \top), (1, 3, 4, \top), (2, 2, 2, \top)\}$.
    Its substitution relation for $x$ is $\{2\}$.
    The row $(1, 3, 4, \top)$ is excluded because the repeated variable $x$ requires the second and third attributes to be equal ($3 \neq 4$), while $(2, 2, 2, \top)$ is excluded because the first attribute does not match the id $1$.\lipicsEnd
\end{example}

\begin{example}[Table-Based Forward Propagation]
    Consider the literal \texttt{(not (= (f x) y))}, where $x$ is restricted to $\{1, 2, 3\}$. If $f$ has a table $(R, 0)$ where $R$ contains the tuples $(1, 10)$ and $(2, 20)$, a left outer join between $X = \{1, 2, 3\}$ and $R$ yields the rows $(1, 10), (2, 20), (3, \text{null})$.  The null value is then replaced by the default $0$. The substitution relation for $(x, y)$ is then $\{(1, 10), (2, 20), (3, 0)\}$.\lipicsEnd
\end{example}

\begin{example}[Clause Substitution Relation]
    Consider the clause \texttt{(or (not (edge x y)) (edge y x))}.
    The substitution relation for literal \texttt{(not (edge x y))} contains the substitutions $(x,y)$ where the literal is not $\top$, i.e., where \texttt{(edge x y)} is not $\bot$. Since \texttt{edge} has table $(R, \bot)$, this is a projection of $R$: \texttt{\{(1,2), (2,3)\}}.
    For literal \texttt{(edge y x)}, the substitution relation contains all substitutions $(y,x)$ for which \texttt{edge y x} is not $\top$. This is the selection $\sigma_{\theta} X_L$, where $\theta$ is the condition that the arguments $(y,x)$ do not match any row in $R$.
    The substitution relation for the clause is the natural join of these two relations: $\{(1,2), (2,3)\} \bowtie \sigma_{\theta} X_L$, i.e., $\{(1,2), (2,3)\}$.\lipicsEnd
\end{example}

\subsubsection{Selective unnesting.}
\label{sec:selective-unnesting}
Recall that a function application $f(\bar t)$ is a nested term when it appears as an argument of a literal or of another function. When a clause contains nested terms, the nesting can be removed by \define{unnesting}: a literal $L[f(\bar t)]$ where $f(\bar t)$ is a nested term is replaced by two literals: $f(\bar t) \neq y \lor L[y]$, where $y$ is a fresh, universally quantified variable. The literal $f(\bar t) \neq y$ is called an \define{unnesting literal}.

The guiding principle of selective unnesting is to avoid unmanageable combinatorial explosions. To this end, a nested term $f(\bar t)$ is unnested only when the resulting unnesting literal $f(\bar t) \neq y$ is a guard in the context of the clause, as described in Sections~\ref{sec:interpreted-symbols} and~\ref{sec:defined-symbols}.

Since these guard conditions can depend on other variables being restricted, unnesting is applied repeatedly until completion, starting with the base guards of the clause. This repetition propagates restrictions both backward and forward.

\begin{example}[Selective Unnesting]
    Consider the literal \texttt{(p (f (g x)))}, where the predicate \texttt{p} has an appropriate table. The nested term \texttt{(f (g x))} is unnested to yield \texttt{(or (not (= (f (g x)) y)) (p y))}, where the flat literal \texttt{(p y)} is a guard that restricts the fresh variable $y$. This unnesting allows the restriction to propagate backward.\lipicsEnd
\end{example}

\subsection{Generalisation to boolean terms}
\label{sec:substitution}

The previous section described how to determine a substitution relation for a clause, which is a specialized formula structure. We now generalize this to arbitrary boolean terms, building on the concepts already defined to handle nested quantifications and complex logical expressions.

The selective unnesting described in Section~\ref{sec:selective-unnesting} is adapted for general formulas following two principles:
\begin{itemize}
    \item the unnesting literal $f(\bar t)=y$ is added to the outermost quantification that binds all variables in $\bar t$, and all occurrences of the function application are replaced by the fresh variable $y$;
    \item if that quantification is universal, say $\forall \bar x: \phi$, the literal is added as a disjunct: $\forall \bar x, y: (f(\bar t) \neq y) \lor \phi$; if it is existential, say $\exists \bar x: \phi$, it is added as a conjunct: $\exists \bar x, y: (f(\bar t) = y) \land \phi$.
\end{itemize}
This approach ensures that the unnesting literal is placed at the highest possible level in the formula structure, allowing restrictions to propagate between its arguments $\bar t$ and the result variable $y$ across as much of the formula as possible.

To take advantage of the duality between existential and universal quantifiers, and to handle negation correctly, we introduce two substitution relations for any boolean term $\phi$:

\begin{definition}
    \label{def:substitution_relation_general}
    The \define{$\exists$ (resp. $\forall$) substitution relation} for a boolean term $\phi$, denoted $S_\phi^{\exists}$ (resp. $S_\phi^{\forall}$), is a relation containing the substitutions of its free variables for which the instantiations of $\phi$ do not simplify to $\bot$ (resp. $\top$) using the look-up tables and background theory laws.
\end{definition}

As the \textit{base case}, for predicates and interpreted theory atoms, the substitution relations are obtained directly from the definitions in Section 3.2: $S_\phi^{\forall}$ is the substitution relation $S_\phi$ for the atom $\phi$, and $S_\phi^{\exists}$ is the substitution relation $S_{\neg \phi}$ for its negation.

As the \textit{inductive case}, the substitution relations for complex terms are computed recursively from their components. The rules for the logical connectives and quantifiers are summarized in Table~\ref{tab:recursive-rules}.

\begin{table}[t!]
    \centering
    \begin{tabular}{Sl Sc Sc Sc }
    \toprule
         & $\phi$ & $S_{\phi}^{\exists}$ & $S_{\phi}^{\forall}$ \\
    \midrule
    1.1 & $\lnot \psi$ & $S_{\psi}^{\forall}$ & $S_{\psi}^{\exists}$ \\
    1.2 & $\psi_0 \land \psi_1$ & $S_{\psi_0}^{\exists} \bowtie S_{\psi_1}^{\exists}$ & $(S_{\psi_0}^{\forall} \times X_{\bar x \setminus \bar x_0}) \cup (S_{\psi_1}^{\forall} \times X_{\bar x \setminus \bar x_1})$ \\
    1.3 & $\psi_0 \lor \psi_1$ & $(S_{\psi_0}^{\exists} \times X_{\bar x \setminus \bar x_0}) \cup (S_{\psi_1}^{\exists} \times X_{\bar x \setminus \bar x_1})$ & $S_{\psi_0}^{\forall} \bowtie S_{\psi_1}^{\forall}$ \\
    1.4 & $\forall \bar y: \psi$ & $\proj{\bar x} S_{\psi}^{\exists}$ & $\proj{\bar x} S_{\psi}^{\forall}$ \\
    1.5 & $\exists \bar y: \psi$ & $\proj{\bar x} S_{\psi}^{\exists}$ & $\proj{\bar x} S_{\psi}^{\forall}$ \\
    \bottomrule
    \end{tabular}
    \caption{Recursive rules for substitution relations. $\bar x$ denotes the free variables of $\phi$, and $\bar x_i$ those of $\psi_i$.}
    \label{tab:recursive-rules}
\end{table}

The fundamental ``key elements'' of this generalized framework are:
\begin{itemize}
    \item \textbf{Negation Duality (Rule 1.1)}: Negation simply swaps the roles of the two relations. This ensures that if a positive formula has a finite $S^\exists$ (e.g., from an appropriate table), its negation will have a finite $S^\forall$, making it an effective guard in a clause.
    \item \textbf{Connective Duality (Rules 1.2, 1.3)}: Conjunction and disjunction are duals. For $S^\exists$, conjunction is a join (restricting possibilities) and disjunction is a union. For $S^\forall$, the roles are reversed. Notably, Rule 1.3 for $S^\forall$ justifies the use of the natural join for clauses in Section 3.2. If the sub-formulas have different free variables, the union is performed after extending the relations with the universal set $X$ for the missing variables.
    \item \textbf{Quantifier Projection (Rules 1.4, 1.5)}: Both universal and existential quantifications reduce to a simple projection of the matrix's relations onto the free variables $\bar x$.
\end{itemize}

These recursive rules allow the grounder to propagate finite restrictions \textbf{upwards} through arbitrary logical structures, even when the underlying universes are infinite.

\begin{example}[Recursive Boolean Propagation]
    To illustrate how finite restrictions propagate through logical connectives, consider the formula $\phi = p(x) \land \neg q(x, y)$. Suppose $p$ has a table $(R_p, \bot)$ and $q$ has a table $(R_q, \top)$.
    By Rule 1.2, $S_\phi^\exists = S_{p(x)}^\exists \bowtie S_{\neg q(x, y)}^\exists$.
    Since the default of $p$ is $\bot$, $S_{p(x)}^\exists = R_p$ is finite.
    By Rule 1.1, $S_{\neg q(x, y)}^\exists = S_{q(x, y)}^\forall$.
    Since the default of $q$ is $\top$, $S_{q(x, y)}^\forall = R_q$ is also finite.
    The resulting relation $S_\phi^\exists = R_p \bowtie R_q$ is finite, effectively restricting both $x$ and $y$.\lipicsEnd
\end{example}

Once these relations are found at the top level, they are used to propagate substitutions back \textbf{downwards} for the final exhaustive expansion, as discussed in the next section.

\subsection{Downward phase: Quantifier elimination}
\label{sec:elimination}

Once the substitution relations are determined via the upward phase, they are used to generate the grounding relation recursively downwards from the root of the formula. Indeed, the actual expansion of a quantification depends on its \define{context}, i.e., the set of substitutions for its free variables already determined by outer formulas.
For a quantification $\forall \bar y : \psi$ (resp. $\exists \bar y : \psi$) in a context $S$, its \define{expansion relation} is the relation $E = S \bowtie S_\psi^\forall$ (resp. $E = S \bowtie S_\psi^\exists$).
It identifies the substitutions for which the quantification must be explicitly expanded.
For any other substitution, either it is not part of the current context $S$, or it is not in the substitution relation ($S_\psi^\forall$ or $S_\psi^\exists$) and thus its matrix simplifies to the quantifier's identity element ($\top$ for $\forall$ and $\bot$ for $\exists$); in both cases, it can be omitted.

The expansion relation provides the mechanism for \define{dynamic guarding}: even if the substitution relation for a sub-formula is infinite (e.g., for a theory comparison like $y > x$), the expansion relations can be finite if the context provides sufficient restrictions.

The following definition and proposition formalize this final step of the method.

\begin{definition}
    \label{def:grounding}
    The \define{grounding relation} $G(\phi, S)$ for a formula $\phi$ in a context $S$ with attributes $\bar x$ is a relation with attributes $\bar x, g_\phi$ that associates each substitution $\bar v \in S$ with a ground formula $g_\phi$ equivalent to the instantiation $\phi[\bar v]$.
\end{definition}

The grounding relation is computed by traversing the structure of $\phi$ recursively.
For literals, $g_\phi$ is the simplified instantiation of the literal for each row, using a left outer join with the table of the predicate or function (for unnesting term) if possible.
For boolean combinations, the grounding relations of the sub-formulas are joined and the logical operator is applied to their grounding attributes. For quantifications, the grounding attribute of the matrix is aggregated over its expansion relation in context $S$ using the appropriate logical operator.

The formal rules for computing the grounding relation are summarized as follows:
\begin{center}
\begin{tabular}{Sc Sc }
    \toprule
         \textbf{Formula $\phi$}
         &  \textbf{Grounding relation $G(\phi, S)$}
         \\ \hline
    \midrule
        literal
          & $\proj{\bar x, g_\phi{:}\phi^*} S$
        \\
        $\diamond \psi_i$
          & $\proj{\bar x, g_\phi{:}\diamond g_{\psi_i}} (\bowtie_i G(\psi_i, S))$
        \\
        $\forall \bar y{:}\psi$
          & $\proj{\bar x, g_\phi{:}\mm{agg}(\land, g_\psi)} G(\psi, S \bowtie S_\psi^\forall)$
        \\
        $\exists \bar y{:}\psi$
          & $\proj{\bar x, g_\phi{:}\mm{agg}(\lor, g_\psi)} G(\psi, S \bowtie S_\psi^\exists)$
        \\
    \bottomrule
\end{tabular}
\end{center}
where $\phi^*$ is the simplified instantiation of the literal, $g_{\psi_i}$ and $g_\psi$ are the grounding attributes for the sub-formulas, and $\diamond$ is a boolean operator.

In practice, this grounding can be computed by a database query to leverage modern query optimizers for best performance.

Notably, for any unnesting literal $f(\bar t) \neq y$ (or $f(\bar t) = y$ for existential quantification) where $f$ is interpreted or has a table, $f(\bar t)$ is a flat term, and all its variables are restricted (possibly by the unnesting itself, as for constructors), the literal simplifies to $\top$ or $\bot$ for every row in the expansion relation. Indeed, as described in Sections~\ref{sec:interpreted-symbols} and~\ref{sec:defined-symbols}, the substitution relation for such a literal only contains substitutions for which the equality $f(\bar t) = y$ holds; thus, during the downward phase, the grounder can evaluate the application and replace the literal by its boolean value, effectively eliminating the interpreted or defined function from the grounded formula.
Conversely, if a literal contains nested terms that were not unnested, its instantiation cannot be reduced to a boolean value.
In that case, the function definition must be communicated to the SMT solver (e.g., as a quantified assertion) to allow the solver to handle the remaining function symbol.

\begin{example}[Grounding and Simplification]
    To illustrate this process, consider a substitution $\{x \mapsto 3, z \mapsto 2, y \mapsto 10\}$ for a clause unnested to \texttt{(not (= (g x) z))} $\lor$ \texttt{(not (= (f z) y))} $\lor$ \texttt{(p y)}. In the downward phase, the grounder produces the ground clause \texttt{(or (not (= (g 3) 2)) (not (= (f 2) 10)) (p 10))}. If the tables for $g$ and $p$ indicate that $g(3)=2$ and $p(10)=\bot$, the literals \texttt{(not (= (g 3) 2))} and \texttt{(p 10)} both simplify to $\bot$. The resulting ground clause then simplifies to \texttt{(not (= (f 2) 10))}.\lipicsEnd
\end{example}

\begin{proposition}
    \label{prop:grounding_correctness}
    The recursive computation described above yields a grounding relation for $\phi$ in context $S$.
\end{proposition}
\begin{proof}
    We must show that for any substitution $\bar v \in S$, the computed attribute $g_\phi$ is equivalent to $\phi[\bar v]$. The proof follows by structural induction on the formula $\phi$:
    \begin{itemize}
        \item if $\phi$ is a literal, $g_\phi$ is defined as the instantiation $\phi^* = \phi[\bar v]$ for each row. The equivalence $\phi[\bar v] \iff g_\phi$ holds trivially;
        \item if $\phi$ is a boolean combination $\diamond \psi_i$, the natural join of the grounding relations $G(\psi_i, S)$ aligns the sub-groundings for the same substitution $\bar v$. By the induction hypothesis, each $g_{\psi_i}$ is equivalent to $\psi_i[\bar v]$; thus, $\diamond g_{\psi_i}$ is equivalent to $(\diamond \psi_i)[\bar v]$;
        \item if $\phi = \forall \bar y{:}\psi$ (resp. $\exists \bar y{:}\psi$), let $E = S \bowtie S_\psi^\forall$ (resp. $E = S \bowtie S_\psi^\exists$) be the expansion relation. By the induction hypothesis, for each $w \in E$, $\psi[w]$ is equivalent to $g_\psi$. For a given $\bar v \in S$, the aggregation $\mm{agg}(\mm{op}, g_\psi)$ over the substitutions $\{w \in E \mid w|_S = \bar v\}$ correctly yields the grounded quantification: the expansion relation ensures that any omitted substitutions for $\bar y$ are those for which $\psi[\bar v, \bar y]$ simplifies to the identity element of the operator $\mm{op}$ (namely $\top$ for $\forall$ and $\bot$ for $\exists$), thereby preserving the truth value of the quantification.
    \end{itemize}
\end{proof}

Hence, one can solve the satisfiability of a sentence by solving the satisfiability of its grounding in the unit context.

\begin{example}[Complete Quantifier Expansion]
    Consider the quantified formula from the introduction:
\begin{lstlisting}
    (forall ((x Int) (y Int))
                (=> (edge x y)
                    (not (= (colorOf x) (colorOf y)))))
\end{lstlisting}
    Its $\forall$ expansion relation in the unit context $1$ is $E = \{(1, 2), (2, 3)\}$. These are the only substitutions for which the matrix does not simplify to $\top$ using the table of \texttt{edge}. The quantification is thus grounded into:
\begin{lstlisting}
    (and (not (= (colorOf 1) (colorOf 2)))
         (not (= (colorOf 2) (colorOf 3))))
\end{lstlisting}
\lipicsEnd
\end{example}

Note that the unions in Rules 1.2 (for $S^\forall$) and 1.3 (for $S^\exists$) can cause redundant instantiations. For example, in $\forall x : (\psi_0(x) \land \psi_1(x))$, the context for the grounding of the matrix is $S^\forall = S_{\psi_0}^\forall \cup S_{\psi_1}^\forall$. Thus, the conjunct $\psi_0(x)$ is needlessly instantiated for $x$ that do not satisfy $\psi_1$ but satisfy $\psi_0$ (and dually for $\psi_1(x)$). This can be avoided by pushing universal quantifiers down conjunctions and existential quantifiers down disjunctions in a preprocessing step, allowing the resulting independent quantifiers to be grounded separately.

\section{Validation}

To validate the quantifier elimination method,
I wrote an implementation of it in a grounder called xmt-lib,
and compared its performance to other approaches.
The benchmarking code is available on-line.\footnote{https://github.com/xmt-lib/benchmarks}

\subsection{Implementation}


The grounder seeks to support the full SMT-LIB 2.6 specification.
However, at this time, it only supports the Core, Int and Real theories of SMT-LIB, and lacks support for the \texttt{match}, \texttt{xor}, \texttt{push} and \texttt{pop} elements of the language.




The result of grounding is a regular SMT-LIB script that can be submitted to any SMT-LIB compliant solver.
For convenience, the xmt-lib program can also execute the script using an embedded Z3 solver.
It can also show the SQL statements used for grounding for debugging purposes.
A demo version is available online.\footnote{https://pcarbonn.github.io/XMT-IDE}

The grounding is performed using a relational database management system (RDBMS) called SQLite.
SQLite is a popular RDBMS that can be easily embedded in other programs.
It has a sophisticated query planner and optimizer\footnote{https://sqlite.org/optoverview.html}.
It has a comprehensive test suite that helps ensure its reliability.
The grounder uses its support of custom functions, including custom aggregation functions.

The grounder is written in Rust for reliability and performance.
It uses two threads: one for the main program and Z3, the other for SQLite.
In the database, logical terms are stored using their textual representations.
Some of the challenges of the implementation are:
\begin{itemize}
    \item the simplification of grounded terms, by handling interpreted operators via custom database functions;
    \item the simplification of relational formula, and the selective addition of the \texttt{DISTINCT} keyword in select statements to avoid duplicate rows (because the semantics of SQL slightly differs from the one of relational theory);
    \item the handling of variable name shadowing.
\end{itemize}

\subsection{Graph coloring}

I first report on experiments on the graph coloring problem presented in the introduction.  I used graphs of various sizes, with a random edge relation that is true for 1 per cent of the pairs of nodes.
I used 4 different encodings of the problem: using either integers or a datatype to identify the nodes, and either a function definition or a quantified assertion to encode the definition of the edge relation.

I observed the time to determine satisfiability, with a time limit of 5 minutes, on a Ubuntu 24.04.4 machine with Intel® Core™ i7-8850H with 12 cores, allocating 10 GB of memory to the task.
The cvc5 solver could not determine the satisfiability (returning ``unknown''), but the Z3 solver could.  Results are shown in Figure~\ref{fig:graph-coloring}.

\begin{figure}[h]
    \centering
    \includegraphics[width=0.9\textwidth]{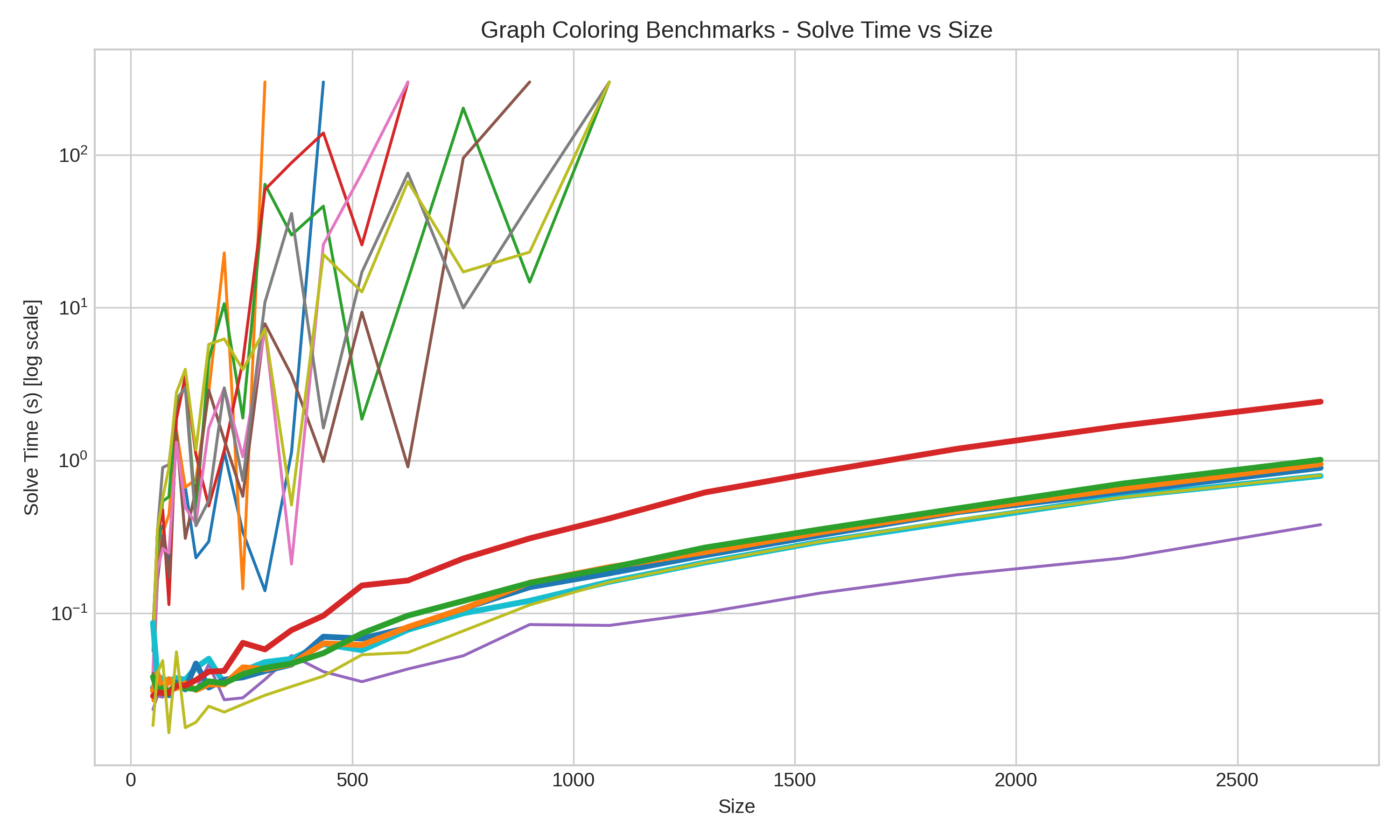}
    \caption{Solve time for different encodings of the graph coloring problem.}
    \label{fig:graph-coloring}
\end{figure}

The lowest line in the Figure is for the quantifier-free encoding solved with Z3: it is the fastest solution, but requires writing a custom program to expand the quantifier.
The lines above it are for the 4 encodings described above solved with `xmt-lib`: it is significantly faster and scales better than solving the same encodings without `xmt-lib`, as shown by the 8 lines above it.  Some of these lines show the performance when E-matching is disabled (which slightly improves performance).

Another preprocessor, called the Ultimate Eliminator\footnote{https://www.ultimate-pa.org/eliminator/}, could not eliminate quantifiers for this problem.

\subsection{SMT-LIB Benchmarks}

I searched for SMT-LIB benchmarks that would benefit from the xmt-lib grounder.  I scanned the non-incremental benchmarks in the 2025 SMT-LIB library to find scripts with defined functions that restrict the range of quantified variables, but found none.


So, I translated benchmarks for SAT and ASP grounders to SMT-LIB, and submitted them for inclusion in the 2026 SMT-LIB library, on the premise that SMT solvers should be able to solve these satisfiability problems reasonably fast.
The submitted benchmarks come from the DIRT benchmark~\cite{DBLP:conf/jelia/LaerVV25}, a collection of benchmarking problems from the literature on grounders.
Three of them are graph problems.

Table~\ref{tab:dirt-results} shows the performance of Z3 4.16.0 and cvc5 1.3.3 on these benchmarks.
The \textbf{xmt} columns show the performance of these solvers when the benchmarks are first preprocessed by \texttt{xmt-lib}.
This is further compared to the performance of clingo 5.8.0~\cite{DBLP:journals/corr/GebserKKS14}, IDP3 3.7.0~\cite{DBLP:journals/corr/CatBBD14} and SLI 0.5.2\footnote{https://gitlab.com/sli-lib/SLI}.

These results show that standalone Z3 and cvc5 have the worst performance on these benchmarks, that \texttt{xmt-lib} significantly improves their performance, and that \texttt{xmt-lib}+Z3 is competitive with SAT and ASP solvers.

\begin{table}[tb]
\centering
\begin{tabular}{l r | r r | r r | r r r}
\toprule
Benchmark & Size & Z3 & xmt & cvc5 & xmt & clingo & IDP3 & SLI \\
\midrule
CommonItems & 200,000 & >60 & 0.63 & >60 & \textbf{0.62} & 6.73 & 0.70 & 1.35 \\
CompleteSets & 200,000 & >60 & 2.71 & >60 & 2.21 & 2.93 & 2.56 & \textbf{1.41} \\
GraphColoring & 2,500 & >60 & 0.68 & >60 & 1.27 & 11.10 & 3.40 & \textbf{0.63} \\
N\_queens & 37 & 4.36 & 0.86 & err & >60 & 1.10 & \textbf{0.58} & 3.71 \\
NPRC & 1000 & >60 & 8.50 & >60 & >60 & \textbf{1.49} & 2.60 & 6.57 \\
PPM & 350 & >60 & >60 & err & >60 & >60 & >60 & >60 \\
PackingProblem & 31 & >60 & 2.04 & err & >60 & >60 & \textbf{0.39} & 16.87 \\
QuasiGroup & 50 & >60 & >60 & err & >60 & \textbf{12.77} & 19.01 & >60 \\
RamseyNumbers & 21 & 5.14 & 1.13 & err & 2.72 & \textbf{0.17} & 0.39 & 0.95 \\
TGCheckSat & 250 & 31.27 & \textbf{0.030} & 3.45 & 0.037 & 0.068 & 0.18 & 0.11 \\
\bottomrule
\end{tabular}
\caption{Results (solve times in seconds) on the DIRT benchmark. err indicates incorrect result or solver error, and the bold font highlights the fastest solver.}
\label{tab:dirt-results}

\end{table}

\section{Related work}

There is a rich body of work on quantifier elimination and complete instantiation for SMT.
Traditional quantifier elimination techniques typically target specific mathematical theories, such as real arithmetic~\cite{Sturm2017} or bit-vectors~\cite{Niemetz2018}, by using theory-specific algebraic or logical manipulations to systematically eliminate quantified variables.
Tools like the Ultimate Eliminator~\cite{ultimate_eliminator} provide practical implementations of these techniques for SMT.
In contrast, our approach generalizes ASP-style relational grounding to SMT, leveraging relational algebra and substitution relations to perform a complete bottom-up expansion, without being strictly tied to the decidability of a specific background theory.

Furthermore, complete instantiation techniques identify finite sets of ground terms that are sufficient to preserve satisfiability for certain fragments of first-order logic~\cite{bradley2006what, ge2009complete, janota2021fair}.
While these methods also address quantification over infinite domains, they operate within the context of lazy SMT solving and typically rely on the solver's internal heuristics and E-matching during the search process.
In contrast, our table-based method performs an eager, complete grounding step by identifying dynamic guards prior to search, thus cleanly separating the grounding phase from the solving phase.

Grounders are a cornerstone of Answer Set Programming (ASP) solvers~\cite{DBLP:conf/aiia/CalimeriFPZ16, DBLP:conf/lpnmr/GebserST07, DBLP:journals/aim/KaufmannLPS16}.
Van Laer et al.~\cite{DBLP:conf/lpnmr/LaerVV24} proposed using bit vectors to ground first-order logic sentences for use by an SMT solver.
However, none of these approaches support quantification over infinite universes.

Wittocx et al.~\cite{DBLP:journals/jair/WittocxMD10} introduced the use of certainly-true and certainly-false bounds for formulas when grounding first-order and FO(ID) sentences for SAT solvers.
These concepts share similarities with the substitution relations presented in Section~3.3; however, their work does not address infinite universes.

Mitchell~\cite{Mitchell2006}, Mohebali et al.~\cite{mohebali2007mxg}, and Aavani et al.~\cite{DBLP:conf/ai/AavaniWTM11} proposed a grounding method for first-order logic formulas with complex terms using relational algebra.
This work builds upon and extends their approach in several key ways:
\begin{itemize}
    \item In their work, the grounding of a formula is constructed bottom-up from the grounding of its sub-terms.
    In contrast, I use substitution relations, which are also constructed bottom-up but then pushed down into quantified formulas.
    This two-way process is essential for grounding quantifications over infinite universes, as it enables the identification of a finite set of relevant variable instantiations for each quantifier.

    \item While their implementation, MXG (now discontinued), reduced formulas to DIMACS (for SAT solvers), \texttt{xmt-lib} targets SMT-LIB.
    Furthermore, \texttt{xmt-lib} supports algebraic data types and arithmetic over the set of reals.
\end{itemize}

\section{Future work}

Other grounding proposal within the theoretical framework could be found to be more efficient than the one proposed here.
Although this work is focused on SMT solvers, the theoretical framework could be adapted
to make proposals for SAT and ASP grounders.
To improve performance for problems with symmetries, one could use a processor before or after the grounder to introduce static symmetry breaking constraints, as is done with ASP solvers~\cite{DBLP:journals/corr/Devriendt016}.
One could also try to incorporate symmetry breaking ideas in a new grounding proposal.
One could also seek to extend the support for inductive definitions.

\bibliography{references}
\end{document}